\documentclass[conference]{IEEEtran}
\ifCLASSINFOpdf
\else
\fi
\usepackage{algorithm}
\usepackage{algorithmic}
\usepackage{amsmath}
\usepackage{amsfonts}	%\mathbb{•}
\usepackage[pdftex]{graphicx}
\usepackage{epstopdf}

\usepackage{nohyperref}  % This makes hyperref commands do nothing without errors
\usepackage{url}  % This makes \url work

\usepackage{morefloats}
\usepackage{ifthen}
\usepackage{subfig}
\usepackage{color}
\newtheorem{theorem}{Theorem}

\newcommand{\figWidth}{0.9\linewidth}	%0.48

\newcommand{\vSpacing}{\vspace*{0.05cm}}
\newcommand{\subHeading}[1]{
\vSpacing
\noindent \textbf{#1} \  
}
\newboolean{short}
\setboolean{short}{true} 

\graphicspath{{figure/}}

\begin{document}
%\title{Bare Demo of IEEEtran.cls for Conferences}
%\title{A distributed low-delay maximal independent set-based scheduling protocol for wireless sensor networks}
\title{A Maximal Concurrency and Low Latency Distributed Scheduling Protocol for Wireless Sensor Networks}

%\author{\IEEEauthorblockN{Michael Shell}
%\IEEEauthorblockA{School of Electrical and\\Computer Engineering\\
%Georgia Institute of Technology\\
%Atlanta, Georgia 30332--0250\\
%Email: http://www.michaelshell.org/contact.html}
%\and
%\IEEEauthorblockN{Homer Simpson}
%\IEEEauthorblockA{Twentieth Century Fox\\
%\and
%\IEEEauthorblockN{James Kirk\\ and Montgomery Scott}
%\IEEEauthorblockA{Starfleet Academy\\
%San Francisco, California 96678-2391\\
\author{\IEEEauthorblockN{Xiaohui Liu}
\IEEEauthorblockA{Department of Computer Science\\
Wayne State University\\
Detroit, MI\\
xiaohui@wayne.edu}
\and
\IEEEauthorblockN{Hongwei Zhang}
\IEEEauthorblockA{Department of Computer Science\\
Wayne State University\\
Detroit, MI\\
hongwei@wayne.edu}}


%Telephone: (800) 555--1212, Fax: (888) 555--1212}


\maketitle

\begin{abstract}
Existing work that schedules concurrent transmissions collision-free suffers from low channel utilization. We propose the Optimal Node Activation Multiple Access (ONAMA) protocol to achieve maximal channel spatial reuse through a distributed maximal independent set (DMIS) algorithm. To overcome DMIS's excessive delay in finding a maximal independent set, we devise a novel technique called pipelined precomputation that decouples DMIS from data transmission. We implement ONAMA on resource-constrained TelosB motes using TinyOS. Extensive measurements on two testbeds independently attest to ONAMA's superb performance 
%in contrast with 
compared to existing work: improving concurrency, throughput, and delay by a factor of 3.7, 3.0, and 5.3, respectively, while still maintaining reliability.
%The abstract goes here.
\end{abstract}


% \fi
%\IEEEpeerreviewmaketitle

%-- Sections go here.
\section{Introduction}

Due to the broadcast nature of wireless communication, access to the shared wireless channel has to be coordinated to avoid interference. There are generally two ways to achieve this, contention-based and TDMA-based. Contention-based protocols are easy to realize, but their performance can be highly dynamic and unpredictable owing to collision and random backoff, especially when node density is high and traffic load is heavy. This explains the prevalence of TDMA-based protocols in scenarios where quality of service (e.g., high throughput, high reliability, and low delay) has to be provided. For instance, the WirelessHART \cite{wirelesshart} and ISA SP100.11a \cite{isa100} standards defined for industrial monitoring and control are both TDMA-based.

%Collision is one of the fundamental hurdles in ensuring QoS in wireless communication. It occurs when multiple packets are transmitted at the same time and get corrupted at the receivers as a result. Collision hurts QoS by decreasing reliability and throughput, increasing delay because of the ensuing retransmissions. They waste the already scarce channel resource. 

%To avoid collision, a series of protocols have been proposed to schedule concurrent transmissions conflict-free in multi-hop wireless networks. 
Numerous TDMA-based protocols have studied how to schedule channel access under interference constraint in multi-hop wireless networks. In the Node Activation Multiple Access protocol (NAMA)\cite{nama:bao:mobicom01}, a node only accesses the channel if its priority is higher than all of its conflicting neighbors', where the priority is a hash function of its unique id and the current time slot. While NAMA ensures no two nodes transmit % or transmits bcoz of "no"??
 simultaneously if they interfere with each other, it can lead to severe concurrency loss and channel underutilization. Figure~\ref{fig:lama_concurrency_loss} illustrates such an example. The number inside each node represents its priority. There is a link between two nodes if their transmissions collide. According to NAMA, only node with priority 7 can be active although nodes with priority 1, 2, 3, and 4 can transmit as well without causing collision.

Given the growing spectrum deficit resulted from ever-increasing demands and fixed amount of spectrum, we aim to squeeze the most capacity out of the limited spectrum. And we propose the Optimal Node Activation Multiple Access (ONAMA) scheduling protocol that activates as many nodes as possible while ensuring collision-free scheduling. It formulates the scheduling problem into finding a maximal independent set in the conflict graph. In graph theory, an independent set is a set of nodes in a graph, none of which are adjacent. An independent set is maximal if adding any other node makes it no longer an independent set. ONAMA includes the distributed maximal independent set (DMIS) algorithm, which identifies a maximal independent set (MIS) of a graph given all node priorities, calculated the same way as NAMA does. %It can activate a node with intermediate priority as long as the node is in the MIS.

\begin{figure}[!tbhp]
\centering
\includegraphics[trim = 0 0 150 550, clip, width=0.7\linewidth]{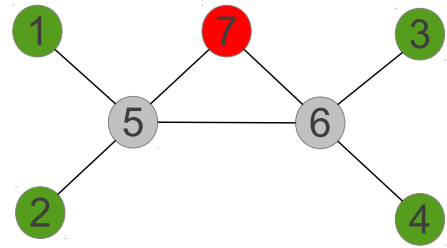}
\caption{NAMA's concurrency loss}
\label{fig:lama_concurrency_loss}
\end{figure}

% when the contention entities are links
% sender-receiver coordination
%When the network, environment, or traffic changes, the overlay conflict graph can also change depending on the interference model. It takes different time for the latest topology to propagate to each node within its conflict set.  Hence, the local contending link set at the sender and receiver of a link can be inconsistent. This is problematic. For instance, in a time slot, the sender, finding the outgoing link to the receiver enjoys the highest priority among its conflict set, declares the link wins the channel and transmits a packet to the receiver. The receiver, whose local conflict set is different from the sender's, may find another outgoing link originating from it wins the channel and it decides to transmit as well. Consequently, collision occurs at the receiver.

% precomputation vs on the fly
%One issue, if left unattended, 
NAMA can compute the schedule for each slot on the fly because it requires no packet exchange. By contrast, it takes multiple rounds of packet exchange for DMIS to find the schedule (i.e., MIS) for a slot. What's worse, the wireless channel is susceptible to packet loss.
%The unreliability of wireless channel exacerbates this situation. 
If ONAMA also computes the schedule on the fly, it introduces significant delays for data delivery, especially in large networks. To reduce delay incurred by MIS computation while activating as many nodes as possible, we decouple the computation of MIS from data transmission by precomputing it in advance. When a certain slot comes, ONAMA simply looks up the precomputed MIS on the fly and activates a node if and only if it is in the MIS.
% pipelining
To further reduce delay, we organize the precomputation of MISs for consecutive slots in a pipeline.
%At any slot, the MIS for the next M slots are being computed simultaneously through pipelining and intermediate results are stored. 

\subHeading{Contributions of this paper.}(1) We develop a distributed scheduling protocol ONAMA that attains maximal activation in a multi-hop wireless network while causing no collision, even if the network is dynamic. It greatly enhances channel spatial reuse % utilization
 compared to the state of the art. (2) We devise a novel technique called pipelined precomputation to address the excessive delay in basic ONAMA. (3) We implement ONAMA on extremely resource-constrained TelosB \cite{telosb} motes using TinyOS \cite{tinyos}. Evaluation in two independent testbeds has verified it increases concurrency by a factor of 3.7, increases throughput by a factor of 3.0, and reduces delay by a factor of 5.3 compared to existing work while maintaining reliability.

\subHeading{Organization of this paper.}Section~\ref{section:protocol} details the ONAMA protocol. We evaluate it comprehensively in Section~\ref{section:evaluation}. Section~\ref{section:relatedwork} discusses related work. We wrap up the article by drawing some conclusions and pointing out some future directions.

\section{The ONAMA Protocol}	\label{section:protocol}
We elaborate on the design and implementation of the ONAMA protocol in this section. First we introduce the baseline version of ONAMA, namely the distributed MIS (DMIS) algorithm, to compute the MIS of a graph in a fully decentralized way. Next, we reduce DMIS's excessive delay by a novel approach called pipelined precomputation. Finally we tackle some challenges in implementing ONAMA on resource-scarce embedded devices.

To unify node- and link-based transmission in a single framework, we construct a corresponding conflict graph on top of the underlying wireless network. In the conflict graph, a node represents a contention entity, a node or a link, in the original network; a link between two nodes exists if data transmissions of the two represented contention entities interfere with each other according to an interference model (e.g., the Physical-Ratio-K model \cite{prk:zhang:secon10}). % introduce PRK model here, pairwise but consider cumulative inteference
Once we acquire scheduling for a conflict graph, it is straightforward to translate it to the corresponding wireless network. From now on, graph is synonymous with conflict graph. We assume time on all nodes is synchronized and divided into slots. This can be attained by, for example, running a time synchronization protocol \cite{prks:zhang:tr}. We also assume each node has a unique id. % and data is transmitted at a fixed power level across the network.

\subsection{Distributed MIS}
Inspired by the observation that NAMA can severely underutilize the channel, we decide to activate as many nodes as possible while still ensuring no two neighboring nodes are active together. In other words, we activate a maximal independent set of a graph in each slot through the following distributed MIS algorithm.
%, a modification of the FastMISv2 algorithm \cite{fastmisv2} adapted to the peculiarities of wireless sensor networks (i.e., limited bandwidth, memory, and computational power).
Even though it resembles some other existing distributed MIS algorithms in appearance, especially the FastMISv2 algorithm \cite{fastmisv2}, it is essentially different from them to factor in the peculiarities of wireless sensor networks, i.e., limited bandwidth, memory, and computational power.

%In ONAMA, time is divided into slots. A node transmits at most one data packet %and multiple control packets, which are shorter, in one slot.
%ONAMA use Time division multiple access (TDMA) with the assistance of a time synchronization protocol such as the Flooding Time Synchronization Protocol (FTSP). 
In DMIS, a node stays in one of three states below at any given time: 
\begin{itemize}
	\item	UNDECIDED:  	it has not decided whether to join the MIS or not.
	\item	ACTIVE: 		it joins the MIS.
	\item	INACTIVE: 	it does not join the MIS.
\end{itemize}
Initially, all nodes are UNDECIDED. In any slot $t$, each node computes the priority of itself and its neighbors according to Equation \ref{equation:hash}
\begin{equation}		\label{equation:hash}
p_i = Hash(i \oplus t) \oplus i.
\end{equation}
$i$ is the node id, $Hash(x)$ is a fast message digest generator that returns a random integer by hashing x, and $p_i$ is i's priority. $\oplus$ concatenates its two operands. Note the second $\oplus$ is necessary to guarantee all nodes' priorities are distinct even when $Hash()$ returns the same number on different inputs. Based on the priority for each node, DMIS computes a MIS for slot t in multiple phases. Each phase consists of three steps:
\begin{enumerate}
  \item Node $v$ exchanges nodal states with its neighbors. \label{enum:state_exchange} % sends its nodal state to its neighbors.
  \item If node $v$'s priority is higher than all its ACTIVE and UNDECIDED neighbors', it enters the MIS by marking itself ACTIVE; conversely, if any of its higher priority neighbors is ACTIVE, it marks itself INACTIVE.	\label{enum:state_transition}
  \item Node $v$ proceeds to the next phase only if its state is UNDECIDED.
\end{enumerate}
% exemplify w/ an evolving diagram?
When no node is UNDECIDED, the algorithm terminates.
Theorem~\ref{theorem:dmisconvergence} shows DMIS does terminate in finite phases.
\begin{theorem}		\label{theorem:dmisconvergence}
DMIS terminates in finite phases.
\end{theorem}
\begin{proof}
Given a graph $G$ and distinct priorities for all its nodes, in the first phase, there are a set of nodes $\mathbb{A}$ going from UNDECIDED to ACTIVE because their priorities are higher than all of their neighbors'. In the following phase, all neighbors of $\mathbb{A}$, denoted as $\mathbb{I}$, become INACTIVE. Removing all nodes in $\left \{ {\mathbb{A} \cup \mathbb{I}} \right\}$ and their adjacent links from $G$, we denote the resulting subgraph as $G'$ and the node set of $G'$ as $V(G')$. No node in $\mathbb{A}$ is adjacent with any node in $V(G')$ because otherwise that node in $V(G')$ would be INACTIVE and in $\mathbb{I}$. So removing all nodes in $\mathbb{A}$ and their adjacent links from $G$ does not affect the ensuing phases. Similarly, removing all nodes in $\mathbb{I}$ and their adjacent links from $G$ does not affect neither since a node becomes ACTIVE or INACTIVE irrespective of its INACTIVE neighbors in phase~\ref{enum:state_transition}. After the first two phases, we only have to run DMIS on the residual subgraph $G'$. After two more phases, $G'$ becomes $G''$. This process continues till the residual subgraph is empty. Because some nodes are removed every time we go from a graph to its subgraph and there are finite nodes in a graph, DMIS terminates in finite phases.
\end{proof}
% Given a graph G and priorities for all its verticies, what's the average distance from any node to its LAMA winner?
% Nearest winner? No. Furthest winner? No.
% decision wave does not propagate at speed of 1 hop/phase? True for line topology, not necessarily for other. Since a node stops if some neighbors are UNDECIDED and none is ACTIVE
Analogous to complexity analysis in \cite{fastmisv2}, it's easy to prove that DMIS terminates in $O(\log(n))$ phases on average, where $n$ is the number of nodes in the graph. Interested readers can refer to \cite{fastmisv2} for more details of the proof. % really??
Theorem~\ref{theorem:dmiscorrect} shows DMIS finds a MIS of the graph.
% hzhang: Instead of an asymptotic analysis, we can have a theorem here showing 1) DMIS does converge to a state of MIS, 2) the exact convergence time. The latter can be used in discussing how to choose M and related tradeoff discussions in following subsections. BTW, "convergence time" is a better word than "running time". A distributed protocol runs infinitely in general and does not have to stop.
\begin{theorem}		\label{theorem:dmiscorrect}
The set $\mathbb{S}$ of all ACTIVE nodes is a MIS of the graph after DMIS terminates.
\end{theorem}
\begin{proof}
We prove the theorem in two steps. 
\begin{enumerate}
	\item $\mathbb{S}$ is an independent set. 
	
	We prove this by contradiction. Suppose $\mathbb{S}$ is not an independent set, then there are two adjacent nodes $u$ and $v$ in $\mathbb{S}$. Since a node only becomes ACTIVE if its priority is higher than all its ACTIVE and UNDECIDED neighbors' and $u$ is ACTIVE, $p_u > p_v$. Likewise, $p_v > p_u$. Contradiction. Thus $\mathbb{S}$ is an independent set.
	\item The independent set $\mathbb{S}$ is maximal. 
	% Again, we prove this by contradiction. Suppose $\mathbb{S}$ is not maximal, then we can add a node $w \not \in \mathbb{S}$ such that $\left \{ {w \cup \mathbb{S}} \right\}$ is still an independent set. 
	
	When DMIS terminates, a node is either ACTIVE or INACTIVE. $\forall u \not \in \mathbb{S}$, it is INACTIVE, which means there exists an ACTIVE neighbor $v \in \mathbb{S}$, whose priority is higher than $u$'s. $\left \{ {u \cup \mathbb{S}} \right\}$ is not independent since $u$ and $v$ are adjacent. Hence $\mathbb{S}$ is a MIS.
\end{enumerate}
\end{proof}

The resulting MIS is the set containing all ACTIVE nodes, which are to be activated in slot t. It is noteworthy that exchanged nodal state does not include priority, which is computed locally even for neighbors'.
%\subsection{Convergence time}
% choice of M
% O(log(n)) proof as in fast MIS??
% worst case: a line sorted by priority -> O(n)
%$M = D * T$
%D: network diameter
%T: depends on node degree and link reliability using control power
%\subsection{Control signalling}
%In step~\ref{enum:state_exchange} of a phase, a node exchange DMIS states with neighbors by sending and receiving control packets, which can piggyback upper layer payload as well if space permits. We further divide each slot into S subslots, out of which one is reserved for data packets and the rest for control packets. Only one data or control packet can be transmitted in each subslot. A node transmits control packets at a power level that can reach all neighbors reliably, which may well be higher than the power level for data transmission. To avoid collision of control packets, each node alternates to transmit control packets in a round-robin fashion with all nodes that interfere with it at its control power level. Interested readers can find detailed discussions on this in \cite{prks:zhang:tr}. Apart from nodal state exchange, control signalling also helps a node tracks all its neighbors and keeps the graph up to date. For this purpose, each node maintains a table containing all its potential neighbors, whose size is L.

\subsection{Pipelined precomputation}
One salient feature of NAMA is that it requires no packet exchange to compute a schedule and can thus be called on the fly. In a TDMA setting, a node can call NAMA as a subroutine at the beginning of a slot to instantaneously determine if it should be active in that slot. The idea of doing the same for DMIS is enticing, but the ability to do so proves elusive. This is because, 
% One naive way to use DMIS is to call it on the fly at the beginning of of a slot to determine if a node should be active in that slot, i.e., if the node is in the resulting MIS. However,
unlike NAMA, running DMIS requires multiple phases and each phase incurs significant delay mainly due to contention to access the shared channel and unreliable channel in step~\ref{enum:state_exchange}. The resulting excessive delay for data delivery is detrimental for a wide range of time-sensitive applications.

To reduce DMIS's delay while retaining its high concurrency, we decouple it from data transmission by precomputing MIS of a slot M slots in advance. 
M is chosen such that DMIS converges within M slots, at least with high probability. In slot t, DMIS starts computing MIS for future slot $(t + M)$ using nodes' priorities at slot $(t + M)$. The intermediate result, i.e., the current MIS, is stored till slot $(t + M)$. When time reaches slot $(t + M)$, a node simply looks up the precomputed MIS and decides to become active or silent, without computing it on the fly like in NAMA.
Since it takes M slots to compute the MIS for each slot, in slot t, MISs for slot $(t + 1), (t + 2), ..., (t + M)$ are being computed. We organize their computation into a pipeline, where the MIS computation of consecutive M slots overlaps. This is more efficient than computing them sequentially. Moreover, we aggregate a vector of M states and exchange them in a single control packet at once, other than convey each of the M states using a separate packet. This greatly saves channel resources.

Figure~\ref{fig:pipelined_precomputation} shows an example of the proposed precomputed pipeline in action when M is 4. The x-axis denotes time in slot, the y-axis denotes the slot whose MIS is being computed. The computation of MIS for slot 4 starts at slot 0 and continues in slots 1, 2, and 3. In slot 4, MIS for this slot has been precomputed and is ready for immediate activation. Similarly, MIS for slot 5 is ready at slot 5, MIS for slot 6 is ready at slot 6, and so on. In slot 3, MISs for the upcoming slots 4, 5, 6, and 7 are being computed simultaneously.

\begin{figure}[!tbhp]
\centering
\includegraphics[trim = 0 0 0 400, clip, width=\figWidth]{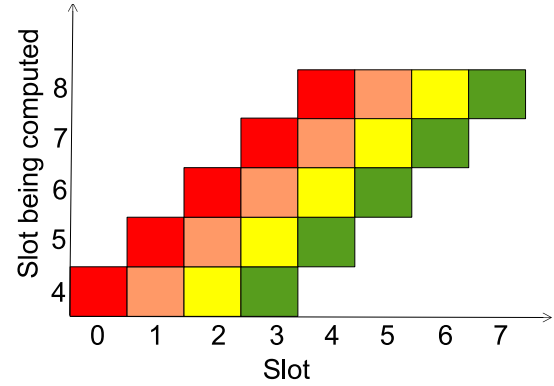}
\caption{Pipelined precomputation}
\label{fig:pipelined_precomputation}
\end{figure}

\subsection{Dynamic graph}
A graph changes over time as nodes join or leave the network, links establish or break. %, due to factors such as time-varying traffic or environment.
This change confuses DMIS because it assumes the graph remains static before it converges. We solve this issue by a snapshot-base approach. Specifically, when starting to compute the MIS for a future slot $(t + M)$ in slot t, we take a snapshot of the graph and use it for the remaining computation even the graph changes within M slots. Hence, the graph is consistent for each call of DMIS. One potential side effect of this approach is that ONAMA defers the usage of the latest graph, which may degrade application performance at upper layer. 
Even if this is the case, the degradation can be mitigated by making M smaller as we shall discuss in subsection~\ref{subsection:impl}.	%??

\subsection{Implementation issues}	\label{subsection:impl}
A typical embedded device is equipped with extremely limited memory. For instance, a TelosB \cite{telosb} mote has only 10 KB of RAM. Implementing ONAMA on such resource-constrained devices poses an additional challenge that is not found on resource-rich ones. To overcome this challenge, we expose several key parameters for fine tuning to let upper layer trade off between memory usage and performance. There are two places in ONAMA that can expend significant amount of memory, especially when the network is large.
\begin{enumerate}
\item
% grouped snapshots
Each node maintains a table containing all its potential neighbors with size L. To store graph snapshot for each slot, a node needs 1 bit for each node in its neighbor table, indicating whether they interfere or not. It thus costs each node $L * M$ bits to store local graph snapshots. To reduce snapshot footprint, we take snapshots every other G slots, instead of every slot. After each snapshot, DMIS uses it to compute the MIS of the next G consecutive slots. This reduces snapshot footprint by a factor of G. Nevertheless, larger G is not always desirable since it  makes the protocol less agile to graph change. G can be tuned to strike a balance between memory consumption and protocol agility.

\item
% multiple ctrl packets in a slot
%Ideally, M can be set arbitrarily large to guarantee distributed MIS convergence. However, each nodes stores $L * M$ intermediate states for precomputation in practice, which implicitly enforces an upper bound for M. 
In step~\ref{enum:state_exchange} of a phase in DMIS, a node exchanges states with neighbors by sending and receiving control packets, which can piggyback upper layer payload as well if space permits. We further divide each slot into S subslots, out of which one is reserved for data packets and the rest for control packets. Only one data or control packet can be transmitted in each subslot. In total, each node stores $L * M$ intermediate states for pipelined precomputation. On the one hand, because a fixed number of control packets are needed for the convergence of DMIS for a slot on a given graph,  a larger S packs more control packets into a slot and thus lessens M and memory expenditure. On the other hand, a larger S also increases control overhead and lowers channel utilization for data delivery. A judicious selection of M again depends on memory consumption and performance tradeoff.
\end{enumerate}

\section{Evaluation} \label{section:evaluation}

\subsection{Methodology}
% Discuss the relation of ONAMA to PRKS here or somewhere else in the paper, and mention how ONAMA is used in PRKS (e.g.,  "link" being treated as "ONAMA" in PRKS).

% integrate (O)LAMA w/ PRKS
% introduces PRK model? No, high-level description suffices
%In the PRK model, a node C interfers with the transmission from another node S to its receiver R if and only if Inequality~\ref{equation:prkmodel} holds:
%\begin{equation}		\label{equation:prkmodel}
%	P(C, R) >= \frac{P(S, R)}{K_{S, R, T_{S, R}}}
%\end{equation}
%It is noteworthy that even though the PRK model defines pairwise interference, it 
%(A, B) interferes with (C, D) if A interferes with (C, D) or C with (A, B).
%
%PRKS uses a control-theoretic approach to instantiating
%the PRK model according to in-situ network and environmental
%conditions, and, through purely local coordination,
%the distributed controllers converge to a state where the desired
%link reliability is guaranteed.
%
%Based on the conflict sets of links,
%data transmissions along individual links can be scheduled in a
%distributed, TDMA manner according to the Link-Activation-
%Multiple-Access (LAMA) algorithm
%
%At the beginning of a time slot, every node executes
%the LAMA scheduling algorithm to decide whether any of its
%associated links will be active in this time slot.

We test ONAMA as a component of the PRK-based scheduling (PRKS) protocol. PRKS \cite{prks:zhang:tr} is a TDMA-based distributed protocol to enable predictable link reliability based on the Physical-Ratio-K (PRK) interference model \cite{prk:zhang:secon10}. The PRK model marries the amenability to distributed protocol design of the ratio-K model (i.e., interference range = K * communication range) and the high fidelity of the signal-to-interference-plus-noise ratio (SINR) model. Through a control-theoretic approach, PRKS instantiates the PRK model parameters according to in-situ network and environment conditions so that each link meets its reliability requirement after convergence. As stated in Section~\ref{section:protocol}, PRKS essentially defines a conflict graph for a given wireless network: a node in the graph denotes a link with data transfer in the network, and a link exists between two nodes in the graph if the corresponding links in the network interference with each other according to the PRK model and its instantiated parameters. Under the condition that link reliability is ensured, PRKS schedules as many nodes as possible in the conflict graph, which is exactly what ONAMA intends to do. Besides ONAMA, PRKS contains many other components such as link estimator, controller, forwarder, time synchronization, and logging. The runtime interactions and resource sharing between ONAMA and them can cause undesirable effects that do not manifest when running ONAMA in isolation. By integrating ONAMA as a part of complex applications like PRKS, we can test the robustness of its design and implementation.

To demonstrate the benefits of ONAMA, we compare the following two PRKS variants:
\begin{itemize}
	\item	PRKS-NAMA: PRKS running on top of NAMA.
	\item	PRKS-ONAMA: PRKS running on top of ONAMA.
\end{itemize}
At the beginning of a time slot, every node calls the NAMA or ONAMA component to decide whether any incident links shall be active in this slot. We measure their performances in two sensor network testbeds, NetEye \cite{neteye} and Indriya \cite{indriya}. 

\subHeading{Network and traffic settings.}
In NetEye and Indriya, we choose each node with probability 0.8 and 0.5, respectively, and the resulting set of nodes forming a random network. For each chosen node A, another node B (also in the random network) is chosen such that the packet delivery ratio (PDR) of the link from A to B is at least 95\% in the absence of interference. For each link, the sender transmits a 128-byte packet to the receiver every 20 ms. Data transmission power is fixed at -25dBm, i.e., power level 3 in TinyOS.
%For more details of the network and traffic settings, please refer to \cite{prks:zhang:tr}.

\subsection{Measurement results}
\subHeading{NetEye testbed.}
For various PDR requirements, Figure~\ref{fig:prks_onama_vs_nama_vs_iorder_concurrency} shows the mean concurrency (i.e., number of concurrent transmissions in a time slot) of PRKS-NAMA and PRKS-ONAMA, as well as a state-of-the-art centralized scheduling protocol iOrder \cite{iorder:che:icnp11}, which also activates as many links as possible while ensuring each link meets its PDR requirement. %, which outperforms well-known interference-oriented scheduling protocols such as Longest-Queue-First \cite{dlqf:shroff:mobihoc10}, GreedyPhysical \cite{greedyphysical:brar:mobicom06}, and LengthDiversity \cite{lengthdiversity:wattenhofer:mobihoc07}. 
Not only does PRKS-ONAMA significantly increase concurrency over PRKS-NAMA by up to 270\%, but also it achieves a concurrency statistically equal or close to that of iOrder despite its distributed nature. This clearly demonstrates the effectiveness of DMIS in ONAMA to activate more links simultaneously.

Figures~\ref{fig:prks_nama_pdr_vs_pdr_req_boxplot} and~\ref{fig:prks_onama_pdr_vs_pdr_req_boxplot} show the boxplots of PDR in PRKS-NAMA and PRKS-ONAMA for different PDR requirements, respectively. We see even though PRKS-ONAMA improves concurrency and channel reuse enormously, it still guarantees that link PDRs meet requirements as PRKS-NAMA does. %As PDR requirements increase, PDRs in both protocols increase to meet them accordingly.

Because of higher concurrency without PDR sacrifice, PRKS-ONAMA's throughput is expected to be higher than PRKS-NAMA's as verified by Figure~\ref{fig:peer_slot_throughput_bar}. Specifically, PRKS-ONAMA's throughput is 3.0, 2.9, 2.7, 2.5 times of PRKS-NAMA's when the PDR requirement is 70\%, 80\%, 90\%, and 95\%, respectively. We also note that as PDR requirement increases, throughputs in both protocols decrease due to lower concurrency.

Figure~\ref{fig:peer_latency_bar} shows PRKS-ONAMA reduces the mean delay of PRKS-NAMA by a factor of 5.3, 4.6, 4.0, and 3.8 when the PDR requirement is 70\%, 80\%, 90\%, and 95\%, respectively. This significant improvement is due to both DMIS and pipelined precomputation in ONAMA.

% 2.30 vs 8.58 vs 8.95
\begin{figure}[!tbhp]
\centering
\includegraphics[width=\figWidth]{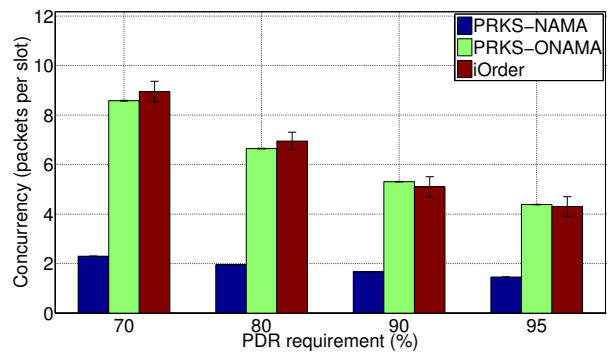}
\caption{Mean concurrency in NetEye}
\label{fig:prks_onama_vs_nama_vs_iorder_concurrency}
\end{figure}

\begin{figure}[!tbhp]
\centering
\includegraphics[width=\figWidth]{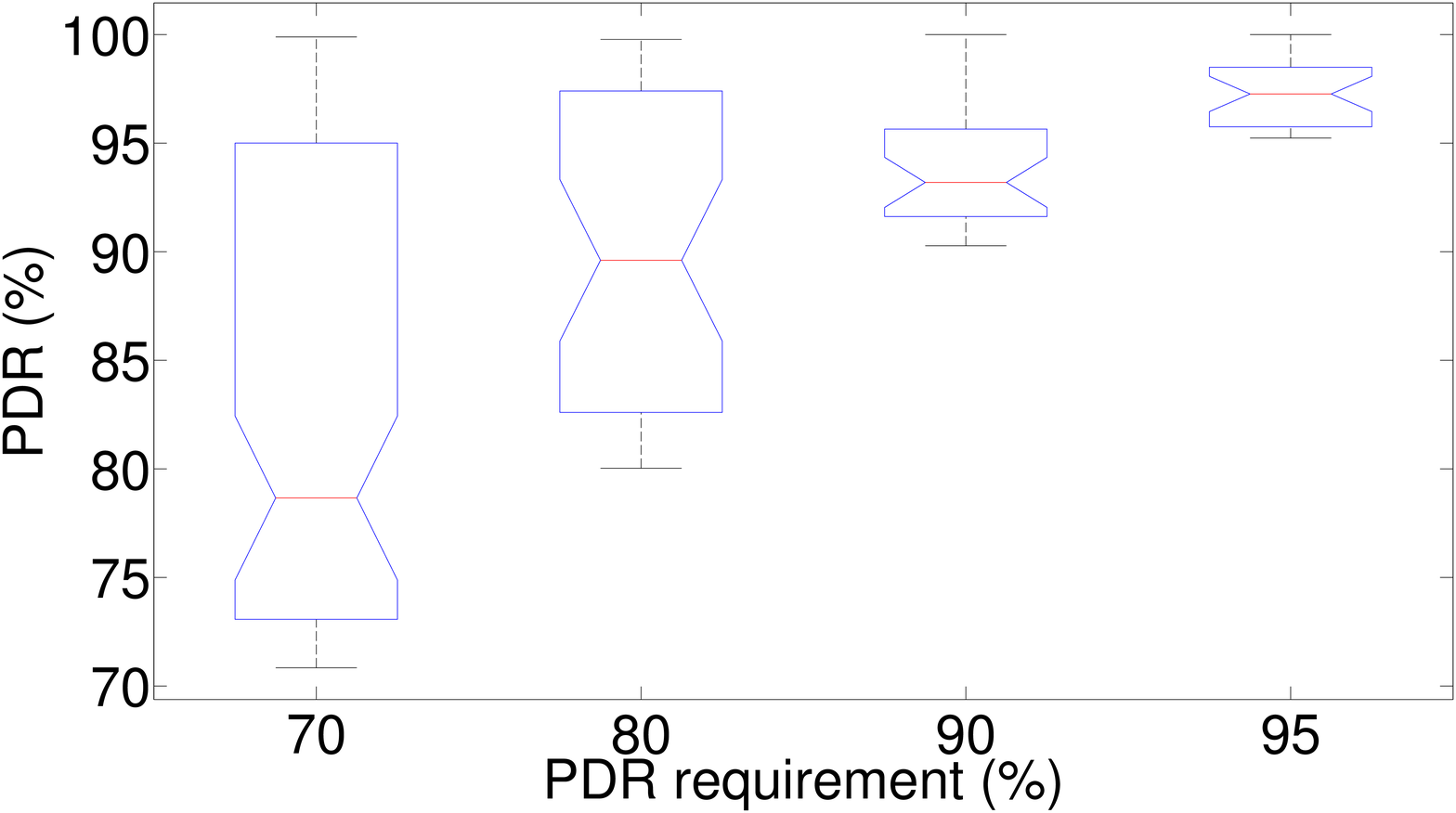}
\caption{Packet delivery ratio (PDR) of PRKS-NAMA in NetEye}
\label{fig:prks_nama_pdr_vs_pdr_req_boxplot}
\end{figure}

\begin{figure}[!tbhp]
\centering
\includegraphics[width=\figWidth]{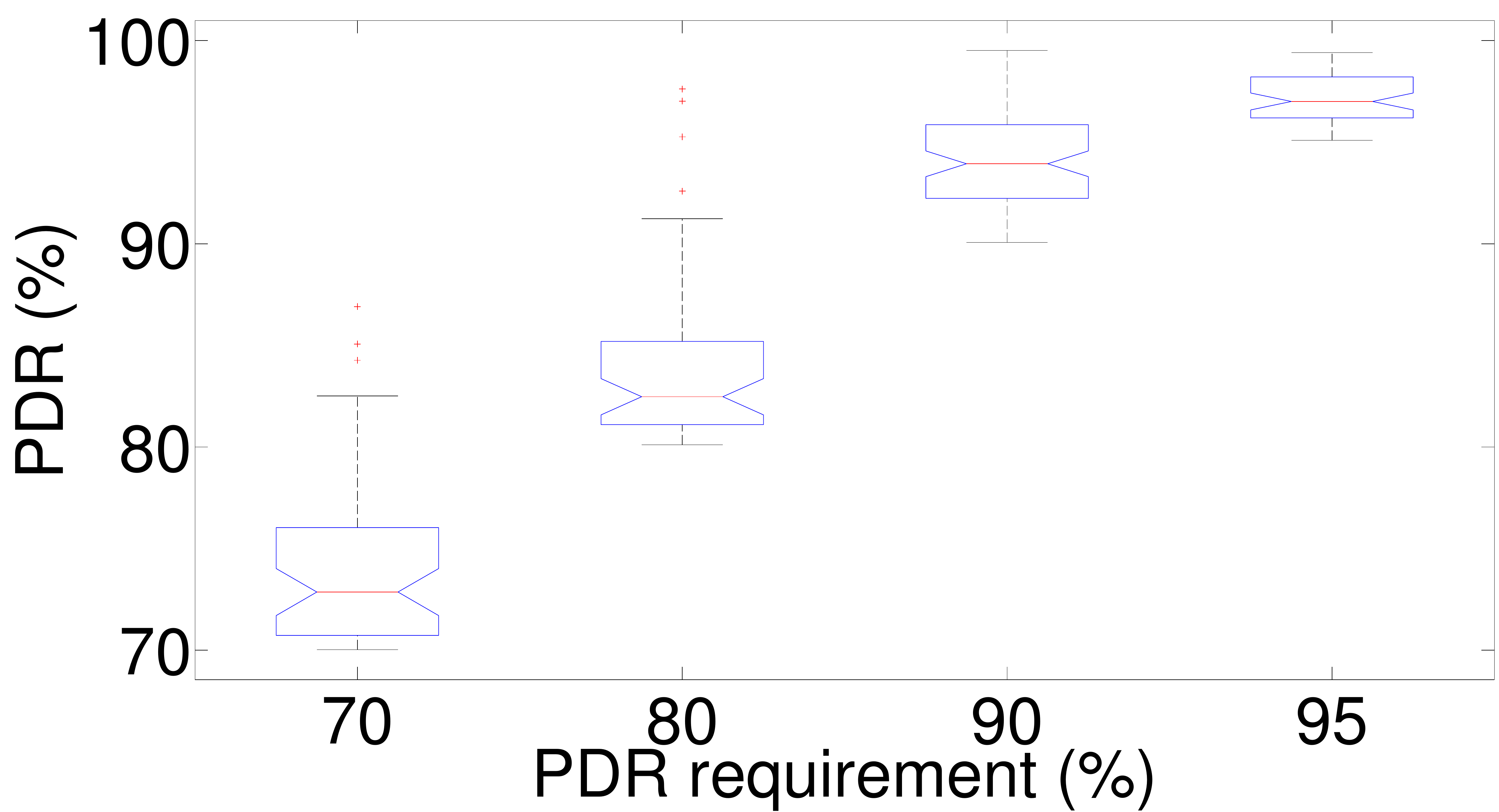}
\caption{Packet delivery ratio (PDR) of PRKS-ONAMA in NetEye}
\label{fig:prks_onama_pdr_vs_pdr_req_boxplot}
\end{figure}

% [358 1083;    340 978;    331 883;    303 750]
% 3.0 2.9 2.7 2.5 times
\begin{figure}[!tbhp]
\centering
\includegraphics[width=\figWidth]{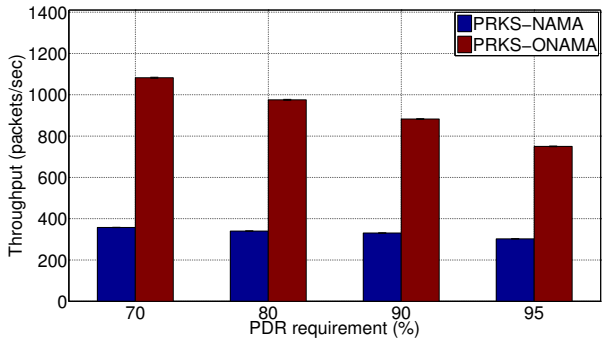}
\caption{Mean throughput in NetEye}
\label{fig:peer_slot_throughput_bar}
\end{figure}

% [282 53; 334 72; 370 92; 432 113]
% 5.3208 4.6389 4.0217 3.8230 times
\begin{figure}[!tbhp]
\centering
\includegraphics[width=\figWidth]{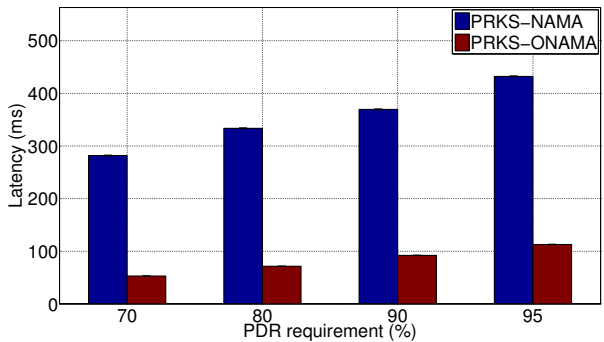}
\caption{Mean delay in NetEye}
\label{fig:peer_latency_bar}
\end{figure}

\subHeading{Indriya testbed.}
Figures~\ref{fig:prks_nama_pdr_vs_pdr_req_indriya_boxplot} and~\ref{fig:prks_onama_pdr_vs_pdr_req_indriya_boxplot} show the link PDRs of PRKS-NAMA and PRKS-ONAMA under different PDR requirements in Indriya, both meeting their requirements as in NetEye.

Figures~\ref{fig:prks_onama_vs_nama_vs_iorder_concurrency_indriya},~\ref{fig:peer_slot_throughput_bar_indriya}, and~\ref{fig:peer_latency_bar_indriya} show the mean concurrency,
%\footnote{Because of the maximal power and Received Signal Strength Indication (RSSI) accuracy constraints in TelosB, we did not run iOrder in Indriya, which is much larger than NetEye.}, 
 mean throughput, and mean delay of both protocols under different PDR requirements in Indriya, respectively. 
% footnote: because of hardware constraint - even power level 31 cannot be heard by everyone else (signal not detected above noise floor) - we skip iOrder in Indriya; also link more time-varying across floors
We see similar relative performance as in NetEye but the degree of improvement is smaller. For instance, PRKS-ONAMA's throughput is only 1.5, 1.3, 1.3, 1.4 times of PRKS-NAMA's when the PDR requirement is 70\%, 80\%, 90\%, and 95\%. This is because links in Indriya are sparser and a larger portion of them are already activated using NAMA.

\begin{figure}[!tbhp]
\centering
\includegraphics[width=\figWidth]{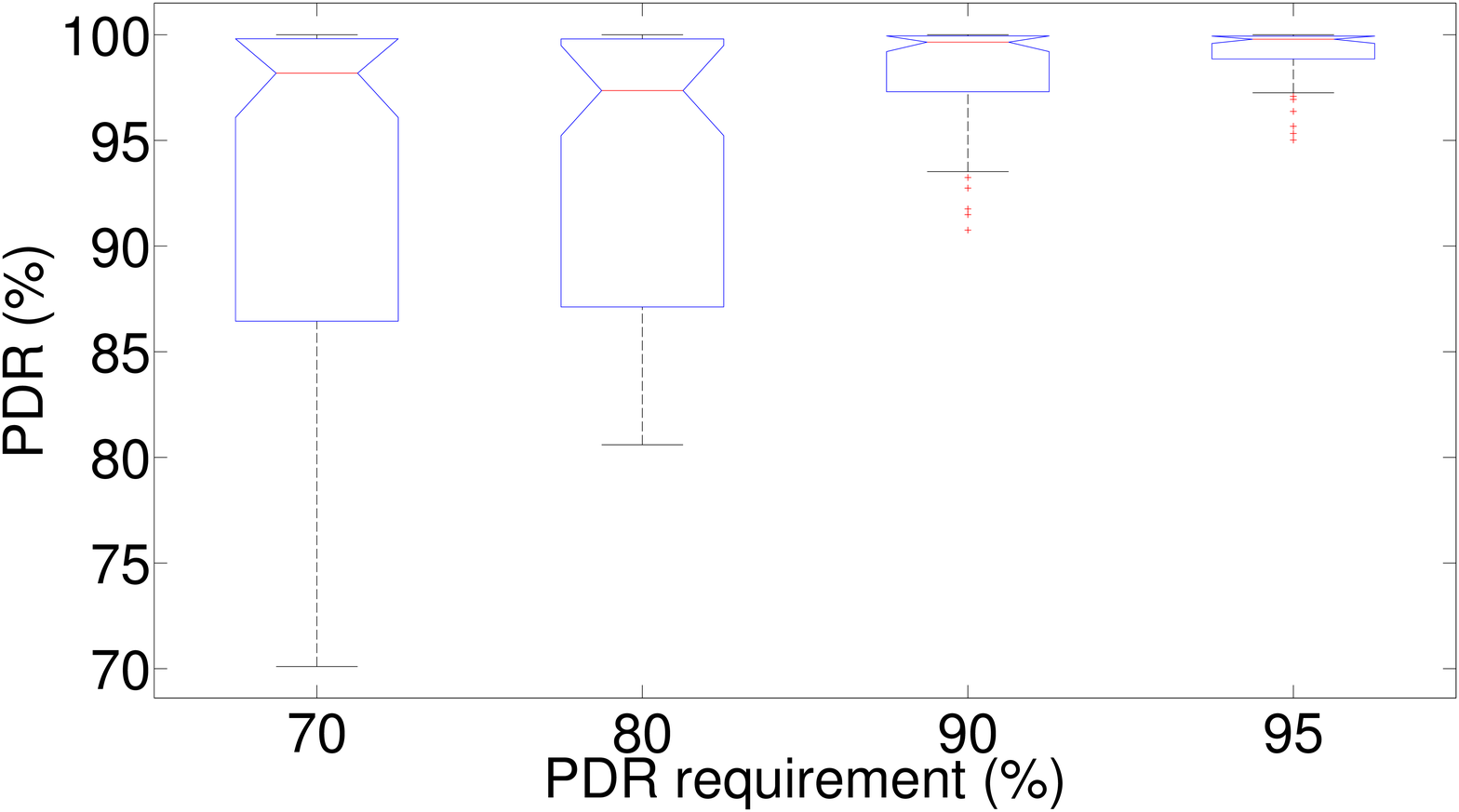}
\caption{Packet delivery ratio (PDR) of PRKS-NAMA in Indriya}
\label{fig:prks_nama_pdr_vs_pdr_req_indriya_boxplot}
\end{figure}

\begin{figure}[!tbhp]
\centering
\includegraphics[width=\figWidth]{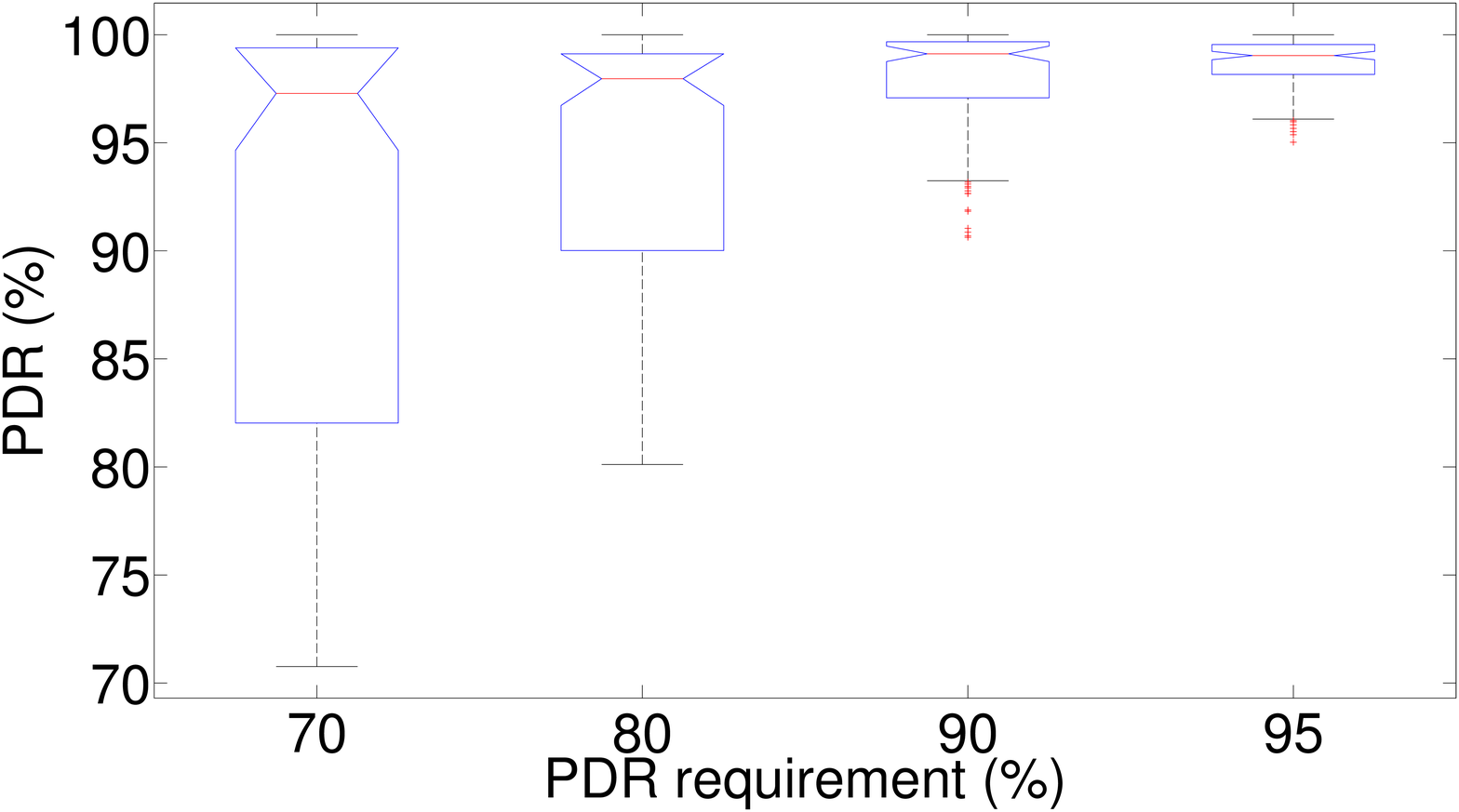}
\caption{Packet delivery ratio (PDR) of PRKS-ONAMA in Indriya}
\label{fig:prks_onama_pdr_vs_pdr_req_indriya_boxplot}
\end{figure}

% [12.3 19.5; 11.5 15.1; 10.5 14.1; 8.7 11.6]
% 1.5854    1.3130    1.3429    1.3333 times
\begin{figure}[!tbhp]
\centering
\includegraphics[width=\figWidth]{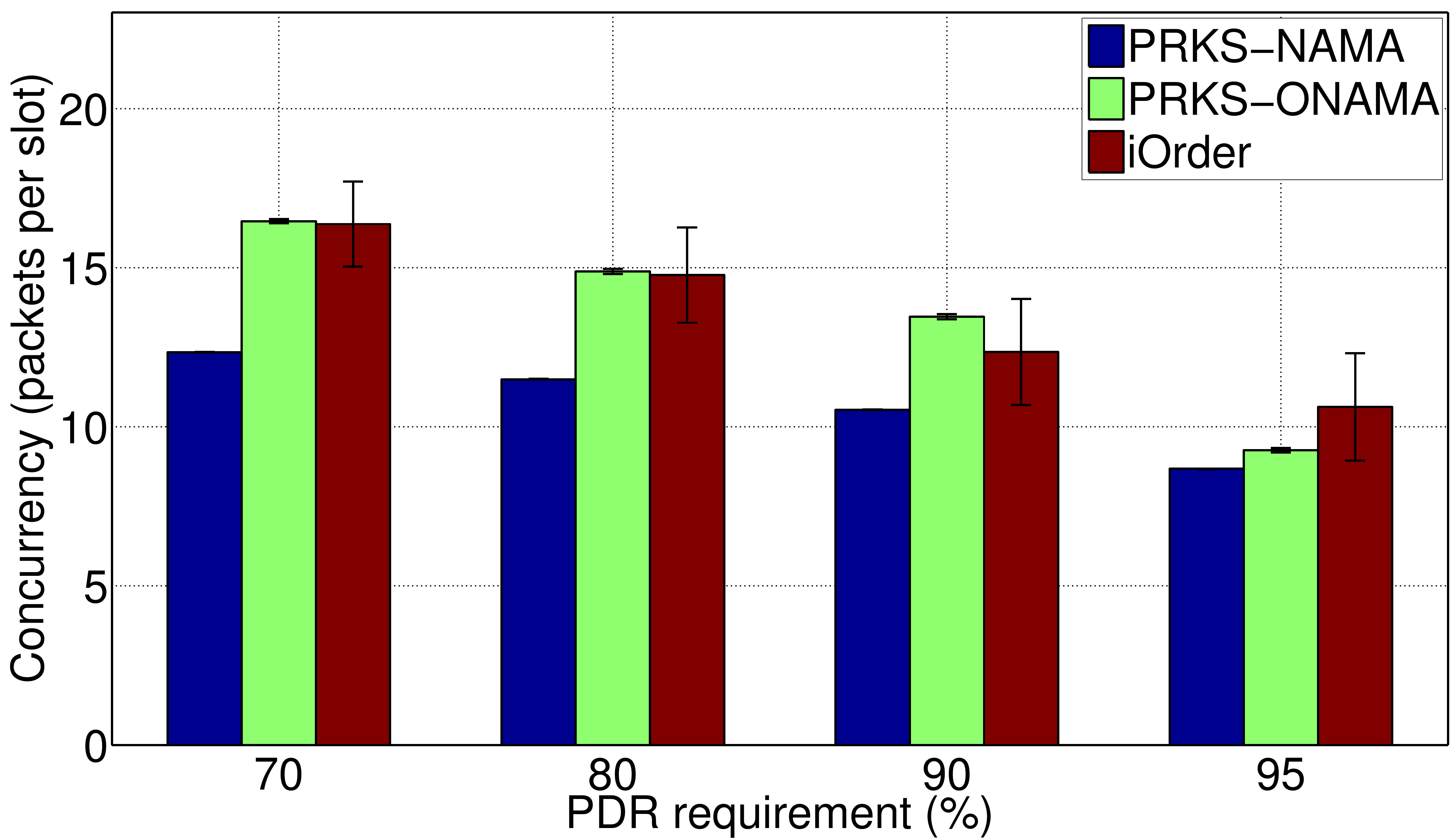}
\caption{Mean concurrency in Indriya}
\label{fig:prks_onama_vs_nama_vs_iorder_concurrency_indriya}
\end{figure}

%[2052 3070; 1947 2573; 1887 2519; 1558 2163]
% 1.4961    1.3215    1.3349    1.3883	times
\begin{figure}[!tbhp]
\centering
\includegraphics[width=\figWidth]{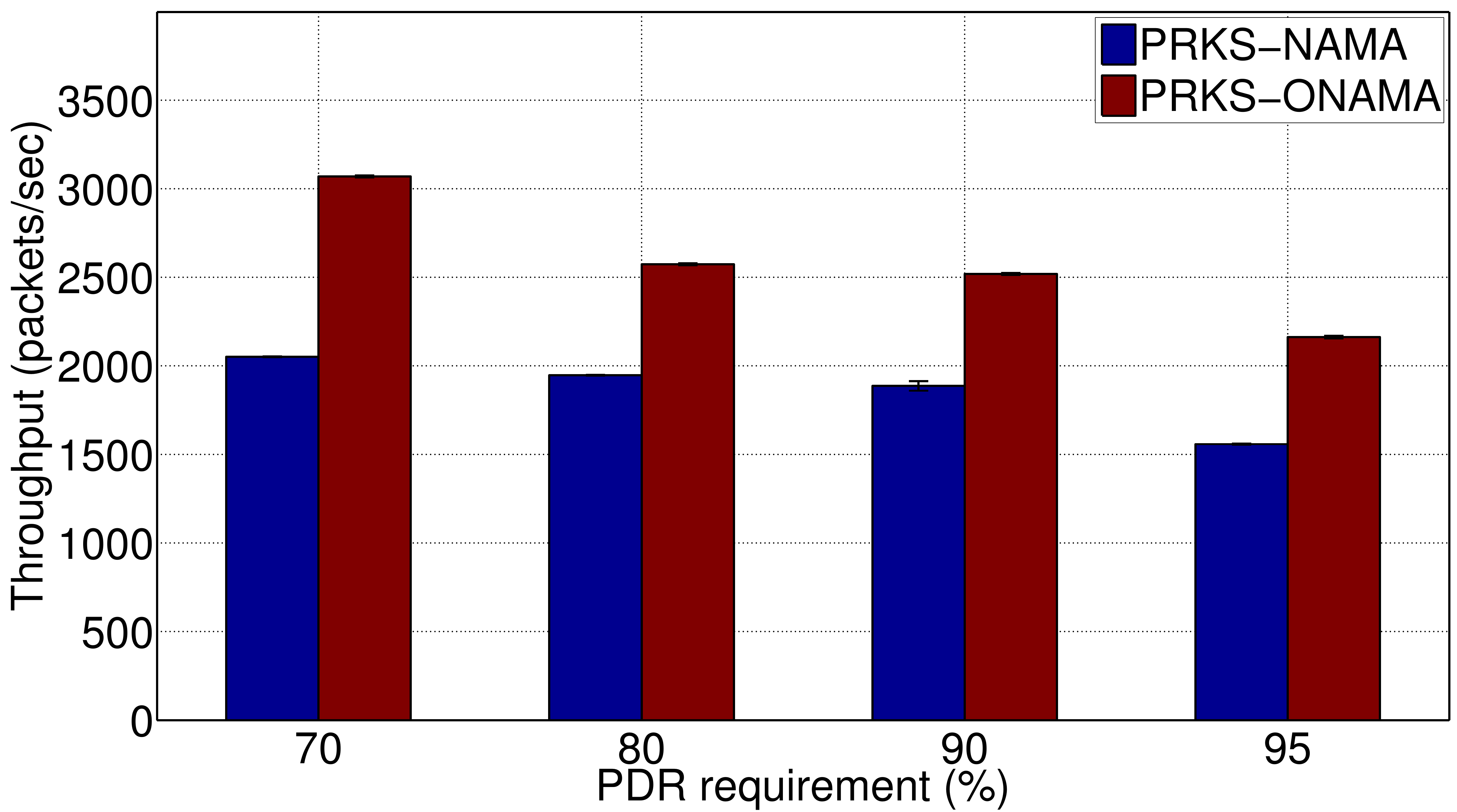}
\caption{Mean network throughput in Indriya}
\label{fig:peer_slot_throughput_bar_indriya}
\end{figure}

% [44.8 15.8; 47.8 19.6; 52.9 21.1; 65.9 25.8]
% 2.8354    2.4388    2.5071    2.5543 times
\begin{figure}[!tbhp]
\centering
\includegraphics[width=\figWidth]{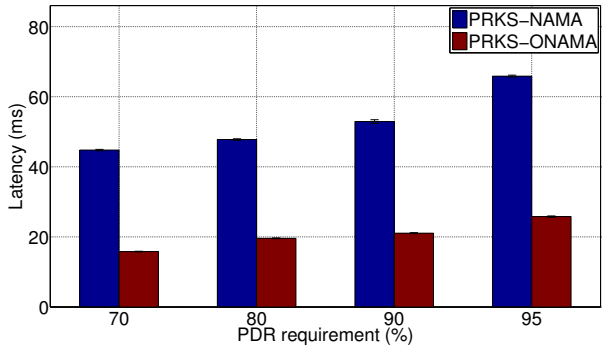}
\caption{Mean delay in Indriya}
\label{fig:peer_latency_bar_indriya}
\end{figure}

\section{Related Work}	\label{section:relatedwork}
There are a plethora of protocols to schedule transmissions collision-free in ad hoc networks, such as NAMA and HAMA, to name a few. In Node Activation Multiple Access protocol (NAMA) \cite{nama:bao:mobicom01}, each node generates a unique priority by hashing and only accesses the channel if its priority is the higher than those of its one-hop and two-hop neighbors. It achieves collision-free scheduling, but suffers from low channel utilization as some nodes are not activated even their activations cause no collision. Hybrid Activation Multiple Access protocol (HAMA) \cite{hama:bao:icnp02} builds upon and improves NAMA by activating non-contending entities, both nodes and links, whenever possible, thus it utilizes channel more efficiently. Unfortunately, it requires radios capable of code division multiplexing. In contrast with HAMA, ONAMA fully utilizes channel without additional requirement on radios.

Orthogonally, plenty of work identifies MIS in a distributed fashion, with the FastMISv2 algorithm in \cite{fastmisv2} being the closest to our work. Akin to DMIS in ONAMA, FastMISv2 also runs in multiple phases; in each phase, FastMISv2 includes a node in the MIS only if its priority is higher than those of all its neighbors, which are excluded from the MIS. There are quite a few critical distinctions, though. (1) Priority is calculated from a pseudorandom number generator (PRNG) with generally different seeds on different nodes, not from a common hash function. A node obtains its neighbor's priority by packet exchange, unlike in DMIS where a node computes its neighbor's priority directly, which expends precious channel resources. Additionally, a new priority is generated from PRNG in each phase, which consumes large amount of MCU time on embedded devices, while DMIS only generates priority once before the first phase. (2) Designed for generic distributed settings, FastMISv2 ignores control signalling and doesn't address the unique challenge posed by unreliable wireless channel. By contrast, ONAMA tackles the challenge explicitly in control signalling by incorporating wireless characteristics. (3) FastMISv2 does not handle the long delay incurred by MIS computation, making it unsuitable for delay-sensitive applications. (4) Lastly but importantly, FastMISv2 assumes a static graph and fails if the graph changes over time.

Finally, unlike all the aforementioned existing work that is evaluated in simulation at best, ONAMA is implemented on resource-limited devices and verified over two real-world testbeds.

\section{Conclusion}		\label{section:conclusion}
In this paper, we propose the ONAMA protocol that schedules maximal number of concurrent transmissions collision-free in a multi-hop wireless network. ONAMA has two pillars: a distributed MIS algorithm tailored to wireless characteristics and the pipelined precomputation technique to reduce DMIS's long delay. Extensive experiments on two testbeds, each with 127+ nodes, have independently shown that it increases concurrency by a factor of 3.7, increases throughput by a factor of 3.0, and reduces delay by a factor of 5.3 compared to the state of the art, while still catering to links' reliability requirements. 
% By using channel resources more efficiently, it also helps PRKS reduce delay and increase throughput.
Not limited to PRKS, ONAMA can be used as a primitive for other applications that require distributed MIS activation. Moreover, pipelined precomputation can be a general technique employed by other distributed applications, where some information is needed imminently but obtaining it requires considerable amount of time.

\bibliographystyle{abbrv}
% for arXiv submission
%\bibliography{../../Programming/LaTex/references}
\bibliography{references}
%\begin{thebibliography}{1}
%\end{thebibliography}

\end{document}